\documentclass[12pt,draftcls,onecolumn, journal]{IEEEtran}
\usepackage{amsmath}
\usepackage{amssymb}
\usepackage{bm}

\usepackage[english]{babel}

\usepackage{graphicx}
\usepackage{psfrag}

\newtheorem{theorem}{Theorem}

\newtheorem{definition}{Definition}
\newtheorem{assumption}{Assumption}

\newtheorem{remarkenv}{Remark}
\newenvironment{remark}{\begin{remarkenv}}{\hfill\raisebox{0.5mm}[0cm][0cm]{$\lhd$}\end{remarkenv}}

\newcommand{\sortbib}[1]{}

\usepackage{dsfont}
\newcommand{\R}{\ensuremath{\mathds{R}}} 

\newcommand{\ones}{\ensuremath{\bm{1}}} 

\newcommand{\dd}{\ensuremath{\mathrm{d}}}                         
\newcommand{\defl}{\ensuremath{\mathrel{\mathop:}=}}              
\newcommand{\defr}{\ensuremath{=\mathrel{\mathop:}}}              
\newcommand{\pd}[2]{\ensuremath{\frac{\partial #1}{\partial #2}}} 
\newcommand{\T}{\ensuremath{\mathrm{T}}}                          

\usepackage{xspace}
\newcommand{\classK}{\ensuremath{\mathcal{K}}\xspace}             

\newcommand{\Acal}{\ensuremath{\mathcal{A}}}

\newcommand{\Ical}{\ensuremath{\mathcal{I}}}
\newcommand{\Jcal}{\ensuremath{\mathcal{J}}}

\newcommand{\Xcal}{\ensuremath{\mathcal{X}}}

\IEEEoverridecommandlockouts
\usepackage{amsfonts}
\usepackage{amssymb}
\usepackage{cite}
\usepackage{xcolor}
\usepackage{bbm}
\usepackage{color}

\newcommand{\qed}{\hfill\blacksquare}
\newtheorem{proposition}{Proposition}
\newtheorem{Corollary}{Corollary}
\newtheorem{Lemma}{Lemma}
\newtheorem{Example}{Example}

\title{Stability Analysis of Monotone Systems via \\ Max-separable {L}yapunov Functions}

\author{H.R.~Feyzmahdavian, B.~Besselink, M.~Johansson%
\thanks{The authors are with the ACCESS Linnaeus Centre and the Department of Automatic Control, School of Electrical Engineering, KTH Royal Institute of Technology, Stockholm, Sweden. Email: \textit{hamidrez@kth.se, bart.besselink@ee.kth.se, mikaelj@kth.se}.}
\thanks{A preliminary version of this work is submitted as the conference paper~\cite{Bart:16}. This manuscript significantly extends the work~\cite{Bart:16} by providing additional technical results and illustrative examples. Namely, \cite{Bart:16} only shows the equivalence of statements $2)$ and $3)$ in Theorem~\ref{thm_omegapath} (rather than the full Theorem~\ref{thm_omegapath}), presents a notion of D-stability that is less general than the notion in the current manuscript, and discusses delay-independent stability of monotone systems with \textit{constant} delays (rather than time-varying and potentially unbounded delays).}
}

\begin{document}
\maketitle
\thispagestyle{empty}
\pagestyle{empty}

%
%

\begin{abstract}
We analyze stability properties of monotone nonlinear systems via max-separable Lyapunov functions, motivated by the following observations: first, recent results have shown that asymptotic stability of a monotone nonlinear system implies the existence of a max-separable Lyapunov function on a compact set; second, for monotone linear systems, asymptotic stability implies the stronger properties of D-stability and insensitivity to time-delays. This paper establishes that for monotone nonlinear systems, equivalence holds between asymptotic stability, the existence of a max-separable Lyapunov function, D-stability, and insensitivity to bounded and unbounded time-varying delays. In particular, a new and general notion of D-stability for monotone nonlinear systems is discussed and a set of necessary and sufficient conditions for delay-independent stability are derived. Examples show how the results extend the  state-of-the-art.
\end{abstract}

%
%

\section{Introduction}
\label{sec:Introduction}

Monotone systems are dynamical systems whose trajectories preserve a partial order relationship on their initial states. Such systems appear naturally in, for example, chemical reaction networks~\cite{Leenheer:07}, consensus dynamics~\cite{moreau_2004}, systems biology~\cite{sontag_2005}, wireless networks~\cite{Yat:95,Feyzmahdavian:12,Boche:08,Feyzmahdavian:14-3}, and as comparison systems in stability analysis of large-scale interconnected systems~\cite{Angeli:07,ruffer_2010,ruffer_2010b}. Due to their wide applicability, monotone systems have attracted considerable attention from the control community  (see, \textit{e.g.},~\cite{Leenheer:01-1,Angeli:03,Aswani:09,Randzer-CDC:14}). Early references on the theory of monotone systems include the papers \cite{Hirsch:82,hirsch_1985,Hirsch:88} by Hirsch and the excellent monograph \cite{book_smith_1995} by Smith.

For monotone \textit{linear} systems (also called positive linear systems), it is known that asymptotic stability of the origin implies further stability properties. First,  asymptotically stable monotone linear systems always admit a Lyapunov function that can be expressed as a weighted max-norm~\cite{book_farina_2000}. Such Lyapunov functions can be written as a maximum of functions with one-dimensional arguments, so they are a particular class of max-separable Lyapunov functions \cite{Ito:12,Ito:14,rantzer_2015}. Second, asymptotic stability of monotone linear systems is robust with respect to scaling of the dynamics with a diagonal matrix, leading to the so-called D-stability property. This notion appeared first in \cite{enthoven_1956} and \cite{arrow_1958}, with additional early results given in~\cite{johnson_1974}. Third, monotone linear systems possess strong robustness properties with respect to time-delays \cite{Haddad:04,Ngoc:06,Buslowicz:08,Kaczorek:09,Liu:09,Liu:10,Liu:2011,Feyzmahdavian:13,briat_2013,Feyzmahdavian:13-0}. Namely, for these systems, asymptotic stability of a time-delay system can be concluded from stability of the corresponding delay-free system, simplifying the analysis.

For monotone \emph{nonlinear} systems, it is in general unknown whether asymptotic stability of the origin implies notions of D-stability or delay-independent stability, even though results exist for certain classes of monotone systems, such as homogeneous and sub-homogeneous systems~\cite{Vahid:11,mason_2009,bokharaie_2010,feyzmahdavian_2014,Feyzmahdavian14:MTNS,Feyzmahdavian:14}. Recent results in \cite{rantzer_2013} and \cite{dirr_2015} show that for monotone nonlinear systems, asymptotic stability of the origin implies the existence of a max-separable Lyapunov function on every compact set in the domain of attraction. Motivated by this result and the strong robustness properties of monotone \emph{linear} systems, we study stability properties of monotone \emph{nonlinear} systems using max-separable Lyapunov functions. Here, monotone nonlinear systems are neither restricted to be homogeneous nor sub-homogeneous.

The main contribution of this paper is to extend the stability properties of monotone linear systems discussed above to monotone nonlinear systems. Specifically, we  demonstrate that asymptotic stability of the origin for a monotone nonlinear system leads to D-stability, and asymptotic stability in the presence of bounded and unbounded time-varying delays. Furthermore, this paper has the following four contributions:

First, we show that for monotone nonlinear systems, the existence of a max-separable Lyapunov function on a compact set is equivalent to the existence of a path in this compact set such that, on this path, the vector field defining the system is negative in all its components. This allows for a simple evaluation of asymptotic stability for such systems.

Second, we define a novel and natural notion of D-stability for monotone nonlinear systems that extends earlier concepts in the literature~\cite{Vahid:11,mason_2009,bokharaie_2010}. Our notion of D-stability is based on the composition of the components of the vector field with an arbitrary monotonically increasing function that plays the role of a scaling. We then show that for monotone nonlinear systems, this notion of D-stability is equivalent to the existence of a max-separable Lyapunov function on a compact set.

Third, we demonstrate that for monotone nonlinear systems, asymptotic stability of the origin is  insensitive to a general class of time-delays which includes bounded and unbounded time-varying delays. Again, this provides an extension of existing results for monotone linear systems to the nonlinear case. In order to impose minimal restrictions on time-delays, our proof technique uses the max-separable Lyapunov function that guarantees asymptotic stability of the origin without delays as a Lyapunov-Razumikhin function.

Fourth, we derive a set of necessary and sufficient conditions for establishing delay-independent stability of monotone nonlinear systems. These conditions can also provide an estimate of the region of attraction for the origin. As in the case of D-stability, we extend several existing results on analysis of monotone systems with time-delays, which often rely on homogeneity and sub-homogeneity of the vector field (see, \textit{e.g.},~\cite{mason_2009,bokharaie_2010,feyzmahdavian_2014,Feyzmahdavian14:MTNS,Feyzmahdavian:14}), to general monotone systems.

The remainder of the paper is organized as follows. Section~\ref{sec_problem statement} reviews some preliminaries on monotone nonlinear systems and max-separable Lyapunov functions, and discusses stability properties of monotone linear systems. In Section~\ref{sec_monotonestab}, our main results for stability properties of monotone nonlinear systems are presented, whereas in Section~\ref{sec_delaysystems}, delay-independent stability conditions for monotone systems with time-varying delays are derived. Section~\ref{sec_discussion} demonstrates through a number of examples how these results extend earlier work in the literature. Finally, conclusions are stated in Section~\ref{sec_conclusions}.

\textit{Notation}.
The set of real numbers is denoted by $\R$, whereas $\R_+ = [0,\infty)$ represents the set of nonnegative real numbers. We let $\R_+^n$ denote the positive orthant in $\R^n$. The associated partial order is given as follows. For vectors $x$ and $y$ in $\R^n$, $x<y$ ($x\leq y$) if and only if $x_i < y_i$ ($x_i\leq y_i$) for all $i\in\Ical_n$, where $x_i\in\R$ represents the $i^{\textup{th}}$ component of~$x$ and $\Ical_n = \{1,\ldots,n\}$. For a real interval~$[a,b]$, $\mathcal{C}\bigl([a,b],\R^{n}\bigr)$ denotes the space of all real-valued continuous functions on $[a,b]$ taking values in $\R^{n}$. A continuous function $\omega:\R_+\rightarrow\R_+$ is said to be of class \classK if  $\omega(0)=0$ and $\omega$ is strictly increasing. A function $\omega:\R^n\rightarrow\R_+$ is called positive definite if $\omega(0)=0$ and $\omega(x)>0$ for all $x \neq 0$. Finally, $\ones_n\in\R^n$ denotes the vector whose components are all one.

%
%

\section{Problem Statement and Preliminaries}
\label{sec_problem statement}

Consider dynamical systems on the positive orthant $\R_+^n$ described by the ordinary differential equation
\begin{align}
\dot{x} = f(x).
\label{eqn_sys}
\end{align}
Here, $x$ is the system state, and the vector field $f:\R_+^n\rightarrow\R^n$ is locally Lipschitz so that local existence and uniqueness of solutions is guaranteed~\cite{book_khalil_2002}. Let $x(t,x_0)$ denote the solution to~\eqref{eqn_sys} starting from the initial condition $x_0\in\R_+^n$ at the time $t\in\R_+$. We further assume that~\eqref{eqn_sys} has an equilibrium point at the origin, \textit{i.e.}, $f(0) = 0$.

\subsection{Preliminaries on Monotone Systems}
In this paper, \textit{monotone systems} will be studied according to the following definition.

\begin{definition}
\label{def_monotonesystem}
The system (\ref{eqn_sys}) is called monotone if the implication
\begin{align}
x'_0 \leq x_0 \;\;\Rightarrow\;\; x(t,x'_0) \leq x(t,x_0), \quad \forall t\in\R_+,
\label{eqn_def_monotonesystem}
\end{align}
holds, for any initial conditions $x_0,x'_0\in \R^n_+$.
\end{definition}

The definition states that trajectories of monotone systems starting at ordered initial conditions preserve the same ordering
during the time evolution. By choosing $x'_0 = 0$ in~(\ref{eqn_def_monotonesystem}), since $x(t,0)=0$ for all $t\in\R_+$, it is easy to see that
\begin{align}
x_0 \in\R_+^n \;\;\Rightarrow\;\; x(t,x_0)\in\R_+^n, \quad \forall t\in\R_+.
\end{align}
This shows that the positive orthant $\R_+^n$ is an invariant set for the monotone system (\ref{eqn_sys}). Thus, monotone systems with an equilibrium point at the origin define positive systems\footnote{A dynamical system given by (\ref{eqn_sys}) is called \textit{positive} if any trajectory of~(\ref{eqn_sys}) starting from nonnegative initial conditions remains forever in the positive orthant, \textit{i.e.}, $x(t)\in\R^n_+$ for all $t\in\R_+$ when $x_0\in\R^n_+$.}.

Monotonicity of dynamical systems is equivalently characterized by the so-called \textit{Kamke condition}, stated next.

\begin{proposition}[\hspace{-0.2mm}\cite{book_smith_1995}]
\label{thm_kamkecondition}
\textit{
The system (\ref{eqn_sys}) is monotone if and only if the following implication holds for all $x,x'\in\R^n_+$ and all $i\in\Ical_n$}:
\begin{align}
x' \leq x\; \textup{and}\; x'_i = x_i \;\;\Rightarrow\;\; f_i(x')\leq f_i(x).
\label{eqn_thm_kamkecondition}
\end{align}
\vskip2mm%
\end{proposition}
Note that if $f$ is continuously differentiable on $\R^n_+$, then condition~\eqref{eqn_thm_kamkecondition} is equivalent to the requirement that $f$ has a Jacobian matrix with nonnegative off-diagonal elements, \textit{i.e.},
\begin{align}
\frac{\partial f_i}{\partial x_j}(x)\geq 0,\quad x \in \R^n_+,
\label{cooperative}
\end{align}
holds for all $i \neq j$, $i,j\in\Ical_n$~\cite[Remark 3.1.1]{book_smith_1995}. A vector field satisfying~\eqref{cooperative} is called \textit{cooperative}.

In this paper, we will consider stability properties of  monotone nonlinear systems. To this end, we use the following definition of (asymptotic) stability, tailored for monotone systems on the positive orthant.

\begin{definition}
\label{def_stability}
The equilibrium point $x=0$ of the monotone system (\ref{eqn_sys}) is said to be stable if, for each $\varepsilon>0$, there exists a $\delta>0$ such that
\begin{align}
0\leq x_0 < \delta\ones_n \;\;\Rightarrow\;\; 0\leq x(t,x_0) < \varepsilon\ones_n, \quad \forall t\in\R_+.
\end{align}
The origin is called asymptotically stable if it is stable and, in addition, $\delta$ can be chosen such that
\begin{align*}
0\leq x_0 < \delta\ones_n \;\;\Rightarrow\;\;\lim_{t\rightarrow\infty}x(t,x_0) = 0.
\end{align*}
\vskip2.5mm
\end{definition}

Note that, due to the equivalence of norms on $\R^n$ and forward invariance of the positive orthant, Definition~\ref{def_stability} is equivalent to the usual notion of Lyapunov stability \cite{book_khalil_2002}.

\subsection{Preliminaries on Max-separable Lyapunov Functions}
We will characterize and study asymptotic stability of monotone nonlinear systems by means of so-called \textit{max-separable Lyapunov functions}
\begin{align}
V(x) = \max_{i\in\Ical_n} \;V_i(x_i),
\label{eqn_maxsepV}
\end{align}
with scalar functions $V_i:\R_+\rightarrow\R_+$. Since the Lyapunov function~(\ref{eqn_maxsepV}) is not necessarily continuously differentiable, we consider its upper-right Dini derivative along solutions of (\ref{eqn_sys}) (see, \textit{e.g.}, \cite{book_rouche_1977}) as
\begin{align}
D^+V(x) = \limsup_{h\rightarrow0^+}\frac{V\bigl(x+hf(x)\bigr) - V\bigl(x\bigr)}{h}.
\label{eqn_Diniderivative}
\end{align}
The following result shows that if the functions $V_i$ in (\ref{eqn_maxsepV}) are continuously differentiable, then~(\ref{eqn_Diniderivative}) admits an explicit expression.

\begin{proposition}[\hspace{-0.2mm}\cite{danskin_1966}]
\label{lem_DVmaxsep}
\textit{
Consider $V:\R_+^n\rightarrow\R_+$ in (\ref{eqn_maxsepV}) and let $V_i:\R_+\rightarrow\R_+$ be continuously differentiable for all $i\in\Ical_n$. Then, the upper-right Dini derivative (\ref{eqn_Diniderivative}) is given by}
\begin{align}
D^+V(x) = \max_{j\in\Jcal(x)} \frac{\partial V_j}{\partial x_j}(x_j)f_j(x),
\end{align}
\textit{where $\Jcal(x)$ is the set of indices for which the maximum in~(\ref{eqn_maxsepV}) is attained}, \textit{i.e.},
\begin{align}
\Jcal(x) = \bigl\{ j\in\Ical_n \;|\; V_j(x_j) = V(x) \bigr\}.
\label{eqn_lem_DVmaxsep_Jcal}
\end{align}
\end{proposition}

\subsection{Preliminaries on Monotone Linear Systems}
Let $f(x) = Ax$ with $A \in \R^{n\times n}$. Then, the nonlinear system (\ref{eqn_sys}) reduces to the linear system
\begin{align}
\dot{x} = Ax.
\label{eqn_syslin}
\end{align}
It is well known that (\ref{eqn_syslin}) is monotone in the sense of Definition~\ref{def_monotonesystem} (and, hence, positive) if and only if $A$ is Metzler, \textit{i.e.}, all off-diagonal elements of $A$ are nonnegative~\cite{book_farina_2000}. We summarize some important stability properties of monotone linear systems in the next result.

\begin{proposition}[\hspace{-0.2mm}\cite{book_berman_1994,rantzer_2015}]
\label{lem_linstab}
\textit{
For the linear system (\ref{eqn_syslin}), suppose that $A$ is Metzler. Then, the following statements are equivalent:
\begin{enumerate}
  \item[1)] The monotone linear system (\ref{eqn_syslin}) is asymptotically stable, \textit{i.e.}, $A$ is Hurwitz.
  \item[2)] There exists a max-separable Lyapunov function of the form
    \begin{align}
      V(x) = \max_{i\in\Ical_n} \:\frac{x_i}{v_i},
      \label{eqn_lem_linstab_maxsepV}
    \end{align}
 on $\R_+^n$, with $v_i>0$ for each $i\in\Ical_n$.
  \item[3)] There exists a vector $w>0$ such that $Aw < 0$.
  \item[4)] For any diagonal matrix $\Delta \in \R^{n\times n}$ with positive diagonal entries, the linear system
\begin{align*}
\dot{x} = \Delta Ax,
\end{align*}
is asymptotically stable, \textit{i.e.}, $\Delta A$ is Hurwitz,
\end{enumerate}
}
\vskip2.5mm
\end{proposition}

In Proposition~\ref{lem_linstab}, the equivalence of statements $1)$ and~$2)$ demonstrates that the existence of a max-separable Lyapunov function is a necessary and sufficient condition for asymptotic stability of monotone linear systems. The positive scalars $v_i$ in the max-separable Lyapunov function (\ref{eqn_lem_linstab_maxsepV}) in the second item can be related to the positive vector $w$ in the third item as $v_i =w_i$ for all $i\in\Ical_n$. Statement $4)$ shows that stability of monotone linear systems is robust with respect to scaling of the rows of matrix $A$. This property is known as \textit{D-stability}~\cite{enthoven_1956}.  Note that the notions of asymptotic stability in Proposition~\ref{lem_linstab} hold globally due to linearity.

Another well-known property of monotone linear systems is that their asymptotic stability is insensitive to bounded and certain classes of unbounded time-delays. This property reads as
follows.

\begin{proposition}[\hspace{-0.2mm}\cite{Sun:12}]\label{prp_lineardelay}
\textit{
Consider the delay-free monotone (positive) system
\begin{align}
\dot{x}(t)=(A+B)x(t),
\label{eqn_syslin_delayfree}
\end{align}
with $A$ Metzler and $B$ having nonnegative elements. If~(\ref{eqn_syslin_delayfree}) is asymptotically stable, then the time-delay linear system
\begin{align}
\dot{x}(t)=Ax(t)+Bx(t-\tau(t))
\label{eqn_syslin_delay}
\end{align}
is  asymptotically stable for all time-varying and potentially unbounded delays satisfying
\begin{align*}
\lim_{t \rightarrow +\infty} t-\tau(t)=+\infty.
\end{align*}
}
 \vskip2mm
\end{proposition}

This results shows that asymptotic stability of the delay-free monotone linear system~(\ref{eqn_syslin_delayfree}) implies that~\eqref{eqn_syslin_delay} is also asymptotically stable. This is a significant property of monotone linear systems, since the introduction of time-delays may, in general, render a stable system unstable~\cite{Gu:03}.

\subsection{Main Goals}

The main objectives of this paper are $(i)$ to derive a counterpart of Proposition~\ref{lem_linstab} for monotone \emph{nonlinear} systems of the form~\eqref{eqn_sys}; and $(ii)$ to extend the delay-independent stability property of monotone \emph{linear} systems stated in Proposition~\ref{prp_lineardelay} to monotone \emph{nonlinear} systems with bounded and unbounded time-varying delays.

%
%

\section{Stability of monotone nonlinear systems}
\label{sec_monotonestab}

The following theorem is our first key result, which establishes a set of necessary and sufficient conditions for asymptotic stability of monotone nonlinear systems.

\begin{theorem}
\label{thm_omegapath}
\textit{
Assume that the nonlinear system (\ref{eqn_sys}) is monotone. Then, the following statements are equivalent:
\begin{enumerate}
\item[1)] The origin is asymptotically stable.
\item[2)] For some compact set of the form
\begin{align}
\Xcal = \bigl\{x\in\R_+^n \;|\; 0\leq x\leq v\bigr \},
\label{eqn_thm_omegapath_Xcal}
\end{align}
with $v>0$, there exists a max-separable Lyapunov function $V:\Xcal\rightarrow\R_+$ as in (\ref{eqn_maxsepV}) with $V_i:[0,v_i]\rightarrow\R_+$ differentiable for each $i\in\Ical_n$ such that
\begin{align}
        \nu_1(x_i) \leq V_i(x_i) \leq \nu_2(x_i),
        \label{eqn_thm_omegapath_Vbounds}
\end{align}
holds for all $x_i\in [0,v_i]$ and for some functions $\nu_1,\nu_2$ of class $\classK$, and that
\begin{align}
        D^+V(x) \leq -\mu(V(x)),
        \label{eqn_thm_omegapath_DV}
\end{align}
holds for all $x\in\Xcal$ and some positive definite function~$\mu$.
\item[3)] For some positive constant $\bar{s}>0$, there exists a function $\rho:[0,\bar{s}]\rightarrow\R_+^n$ with $\rho_i$ of class $\classK$, $\smash{\rho_i^{-1}}$ differentiable on $[0,\rho_i(\bar{s})]$ and satisfying
\begin{align}
        \frac{\dd \rho_i^{-1}}{\dd s}(s) > 0,
        \label{eqn_thm_omegapath_drhods}
\end{align}
      for all $s\in\bigl(0,\rho_i(\bar{s})\bigr]$ and  all $i\in\Ical_n$, such that
\begin{align}
        f\circ\rho(s) \leq -\alpha(s),
        \label{eqn_thm_omegapath_rhoineq}
\end{align}
      holds for $s\in [0,\bar{s}]$ and some function $\alpha:[0,\bar{s}]\rightarrow\R_+^n$ with $\alpha_i$ positive definite for all $i \in \Ical_n$.
\item[4)] For any function $\psi:\R^n_+ \times \R^n\ \rightarrow\R^n$ given by $\psi(x,y) =\bigl(\begin{array}{ccc}\psi_1(x_1,y_1), &\ldots&, \psi_n(x_n,y_n)\end{array}\bigr)^{\T}$ where
\begin{itemize}
\item $\psi_i:\R_+ \times \R\ \rightarrow\R$ for $i \in \Ical_n$,
\item $\psi_i(x_i,0) = 0$ for any $x_i\in \R_+$ and all $i \in \Ical_n$,
\item $\psi_i(x_i,y_i)$ is monotonically increasing in $y_i$ for each nonzero $x_i$, \textit{i.e.}, the implication
\begin{align}
y_i' < y_i \;\;\Rightarrow\;\; \psi_i(x_i,y_i') < \psi_i(x_i,y_i)
\label{eqn_thm_dstab_phimonotone}
\end{align}
holds for any $x_i>0$ and all $i\in\Ical_n$,
\end{itemize}
the nonlinear system
\begin{align}
\dot{x} = \psi\bigl(x, f(x)\bigr),
\label{eqn_thm_dstab_sysphi}
\end{align}
has an asymptotically stable equilibrium point at the origin.
\end{enumerate}
}
\end{theorem}

\begin{proof}
The proof is given in Appendix~A.
\end{proof}

Theorem~\ref{thm_omegapath} can be regarded as a nonlinear counterpart of Proposition~\ref{lem_linstab}. Namely, choosing the functions $V_i$ in the second statement of Theorem~\ref{thm_omegapath} as $V_i(x_i) = x_i/v_i$, $v_i>0$, and the function $\rho$ in the third statement as $\rho(s) = ws$, $w>0$, recovers statements $2)$ and~$3)$ in Proposition~\ref{lem_linstab}, respectively. Note that we can let $v_i=w_i$ for each $i \in \Ical_n$ since, according to the proof of Theorem~\ref{thm_omegapath}, the relation between $V$ and $\rho$ is
\begin{align*}
V_i(x_i) = \rho_i^{-1}(x_i)\;\;\textup{and}\;\;\rho_i(s) = V_i^{-1}(s),\;i\in\Ical_n.
\end{align*}
Statement $4)$ of Theorem~\ref{thm_omegapath}, which can be regarded as a notion of D-stability for nonlinear monotone systems, is a counterpart of the fourth statement of Proposition~\ref{lem_linstab}. More precisely, the choice $\psi(x,y) = \Delta y$ for some diagonal matrix $\Delta$ with positive diagonal entries recovers the corresponding result for monotone linear systems in Proposition~\ref{lem_linstab}.

\begin{remark}
\label{remark_region}
According to the proof of Theorem~\ref{thm_omegapath}, if there is a function $\rho$ satisfying the third statement for $s\in[0,\bar{s}]$, then the Lyapunov function (\ref{eqn_maxsepV}) with components $V_i(x_i) = \rho_i^{-1}(x_i)$ guarantees asymptotic stability of the origin for any initial condition
\begin{align*}
x_0\in \Xcal = \bigl\{x\in\R_+^n \;|\; 0\leq x\leq \rho(\bar{s})\bigr \}.
\end{align*}
This means that the set $\Xcal$ is an estimate of the region of attraction for the origin.
\end{remark}

We now present a simple example to illustrate the use of Theorem~\ref{thm_omegapath}.

\begin{Example}
\label{nondelayedExample}
Consider the nonlinear dynamical system
\begin{align}
\dot{x}= f(x)= \begin{bmatrix} -5x_1+x_1x_2^2 \\  x_1- 2x_2^2 \end{bmatrix}.
\label{Example1}
\end{align}
This system has an equilibrium point at the origin. As the Jacobian matrix of $f$ is Metzler for all $(x_1,x_2)\in\R^2_+$, $f$ is cooperative. Thus, according to Proposition~\ref{thm_kamkecondition},~\eqref{Example1} is monotone on $\R^2_+$.

First, we will show that the origin is asymptotically stable. Let $\rho(s)=(s, \sqrt{s})$, $s\in [0,4]$. For each $i\in\{1,2\}$, $\rho_i$ is of class $\classK$. It is easy to verify that
\begin{align*}
\smash{\rho_1^{-1}}(s)=s,\;\;\smash{\rho_2^{-1}}(s)=s^2.
\end{align*}
Thus, $\smash{\rho_i^{-1}}$, $i\in\{1,2\}$, is continuously differentiable and satisfies~\eqref{eqn_thm_omegapath_drhods}. In addition,
\begin{align*}
 f\circ\rho(s) = \begin{bmatrix} -5s+s^2 \\  - s \end{bmatrix}\leq - \begin{bmatrix} s \\   s \end{bmatrix},
\end{align*}
for all $s\in [0,4]$, which implies that \eqref{eqn_thm_omegapath_rhoineq} holds. It follows from the equivalence  of statements $1)$ and $3)$ in Theorem~\ref{thm_omegapath} that the origin is asymptotically stable.

Next, we will estimate the region of attraction of the origin by constructing a max-separable Lyapunov function. According to Remark~\ref{remark_region}, the monotone system~\eqref{Example1} admits the max-separable Lyapunov function
\begin{align*}
V(x)=\max\left\{x_1,x_2^2\right\}
\end{align*}
that guarantees the origin is asymptotically stable for 
\begin{align*}
x_0\in \bigl\{x\in\R_+^2 \;|\; 0\leq x\leq (4,2)\bigr \}.
\end{align*}

Finally, we discuss D-stability of the monotone system~\eqref{Example1}. Consider $\psi$ given by
\begin{align*}
\psi(x,y)= \left(\begin{array}{cc} \frac{x_1}{x_1^2+1}y_1^3\;, & x_2^2 y_2 \end{array}\right)^{\T}.
\end{align*}
For any $x\in\R^2_+$, $\psi(x,0)=0$. Moreover, each component $\psi_i(x_i,y_i)$, $i\in\{1,2\}$, is monotonically increasing for any $x_i>0$. Since the origin is an asymptotically stable equilibrium point of~\eqref{Example1}, by the equivalence  of statements $1)$ and $4)$ in Theorem~\ref{thm_omegapath}, the monotone nonlinear system
\begin{align*}
\dot{x} = \psi\bigl(x,f(x)\bigr)=\begin{bmatrix} \frac{x_1}{x_1^2+1}\left(-5x_1+x_1x_2^2\right)^3 \\ x_2^2\left( x_1- 2x_2^2\right) \end{bmatrix}
\end{align*}
has an asymptotically stable equilibrium point at the origin.
\end{Example}

\begin{remark}
\label{remark_equ}
A consequence of the proof of Theorem~\ref{thm_omegapath} is that all statements $1)$--$4)$ are equivalent to the existence of a vector $w>0$ such that $f(w)<0$ and
\begin{align}
\lim_{t\rightarrow \infty} x(t,w)=0.
\label{extra-stability-assumption}
\end{align}
Contrary to statement $3)$ of Proposition~\ref{lem_linstab} for monotone linear systems, the condition $f(w)<0$ without the additional assumption~\eqref{extra-stability-assumption} does not necessarily guarantee asymptotic stability of the origin for monotone nonlinear systems. To illustrate the point, consider, for example, a scalar monotone system described by (\ref{eqn_sys}) with $f(x)=-x(x-1)$, $x\in\R_+$. This system has two equilibrium points:  $x^\star=0$ and $x^{\star}=1$. Although $f(2)<0$, it is easy to verify that any trajectory starting from the initial condition $x_0>0$ converges to $x^\star=1$. Hence, the origin is not stable.
\end{remark}

\begin{remark}
Several implications in Theorem~\ref{thm_omegapath} are based on similar results in the literature. Namely, the implication $1)\Rightarrow 2)$ was shown in \cite{rantzer_2013} (see also \cite{dirr_2015}) by construction of a max-separable Lyapunov function that is not necessarily differentiable or even continuous (see \cite[Example~2]{dirr_2015}). The implication $3) \Rightarrow 2)$ was proven before in~\cite[Theorem~III.2]{ruffer_2010c} for the case of global asymptotic stability of monotone nonlinear systems considering a max-separable Lyapunov function with possibly non-smooth components $V_i$. The implication $1) \Rightarrow 4)$ was shown in \cite{mason_2009,bokharaie_2010,Vahid:11} for particular classes of scaling function $\psi$. For example,  if we choose $\psi_i(x_i,y_i)=d_i(x_i)y_i$ with $d_i(x_i)>0$ for $x_i>0$, then statement~$4)$ recovers the results in~\cite{bokharaie_2010,Vahid:11}. However, contrary to~\cite{mason_2009,bokharaie_2010,Vahid:11}, neither homogeneity nor sub-homogeneity of $f$ is required in Theorem~\ref{thm_omegapath}.
\end{remark}

%
%

\section{Stability of monotone systems with delays}
\label{sec_delaysystems}

In this section, delay-independent stability of nonlinear systems of the
form
\begin{align}
& \left\{
\begin{array}{rll}
\dot{x}(t) &= g\bigl(x(t),x(t-\tau(t))\bigr)\!&,\; t\geq 0,\\
x(t) &= \varphi(t)\vphantom{\bigl(} &,\; t\in[-\tau_{\max},0],
\end{array}
\right.
\label{Delaysystem}
\end{align}
is considered. Here, $g:\R^n_+ \times \R^n_+ \rightarrow \R^n$ is locally Lipschitz continuous with $g(0,0)=0$, $\varphi\in \mathcal{C}\bigl([-\tau_{\max},0],\R_+^n\bigr)$ is the vector-valued function specifying the initial state of the system, and  $\tau$ is the time-varying delay which satisfies the following assumption:

\begin{assumption}
\label{General Delays Assumption}
The delay $\tau:\R_+\rightarrow \R_+$ is continuous with respect to time and satisfies
\begin{align}
\label{General Delays Condition}
\lim_{t \rightarrow +\infty} t-\tau(t)=+\infty.
\end{align}
\vskip2mm%
\end{assumption}

Note that $\tau$ is not necessarily continuously differentiable and that no restriction on its derivative (such as $\dot{\tau}(t)<1$) is imposed. Roughly speaking, condition~\eqref{General Delays Condition} implies that as $t$ increases, the delay $\tau(t)$ grows slower than time itself. It is easy to verify that all bounded delays, irrespectively of whether they are constant or time-varying, satisfy Assumption~\ref{General Delays Assumption}. Moreover, delays satisfying~\eqref{General Delays Condition} may be unbounded (take, for example, $\tau(t)=\gamma t$ with $\gamma \in (0,1)$).

Unlike the non-delayed system~\eqref{eqn_sys}, the solution of the time-delay system~\eqref{Delaysystem} is not uniquely determined by a point-wise initial condition $x_0$, but by the continuous function $\varphi$ defined over the interval $[-\tau_{\max},0]$. Assumption~\ref{General Delays Assumption} implies that there is a sufficiently large $T>0$ such that $t-\tau(t)>0$ for all $t>T$. Define
\begin{align*}
\tau_{\max}=-\inf_{0\leq t\leq T}\biggl\{t-\tau(t)\biggr\}.
\end{align*}
Clearly, $\tau_{\max}\in\R_+$ is bounded ($\tau_{\max}<+\infty$). Therefore, the initial condition $\varphi$ is defined on a bounded set $[-\tau_{\max},0]$ for any delay satisfying Assumption~\ref{General Delays Assumption}, even if it is unbounded. Since $g$ is Lipschitz continuous and $\tau$ is a continuous function of time, the existence and uniqueness of solutions to~\eqref{Delaysystem} follow from \cite[Theorem~2]{driver_1962}. We denote the solution to~\eqref{Delaysystem} corresponding to the initial condition $\varphi$ by $x(t,\varphi)$.

From this point on, it is assumed that the time-delay system~\eqref{Delaysystem} satisfies the next assumption:

\begin{assumption}
\label{Assumption_delaysystem}
The following properties hold:
\begin{enumerate}
\item $g(x,y)$ satisfies Kamke condition in $x$ for each $y$, \textit{i.e.},
\begin{align}
x'\leq x\; \textup{and}\;x'_i = x_i \;\;\Rightarrow\;\; g_i(x',y)\leq g_i(x,y),
\end{align}
holds for any $y \in \R^n_+$ and all $i\in\Ical_n$.
\item $g(x,y)$ is order-preserving in $y$ for each $x$, \textit{i.e.},
\begin{align}
y'\leq y \;\;\Rightarrow\;\; g(x,y')\leq g(x,y),
\end{align}
holds for any $x \in \R^n_+ $.
\end{enumerate}
\end{assumption}

System~\eqref{Delaysystem} is called monotone if given two initial conditions  $\varphi,\varphi'\in \mathcal{C}\bigl([-\tau_{\max},0],\R_+^n\bigr)$ with $\varphi'(t)\leq \varphi(t)$ for all $t\in[-\tau_{\max},0]$, then
\begin{align}
x(t,\varphi')\leq x(t,\varphi),\quad \forall t\in \R_+.
\end{align}
It follows from \cite[Theorem 5.1.1]{book_smith_1995} that Assumption~\ref{Assumption_delaysystem} ensures the monotonicity of~\eqref{Delaysystem}. Furthermore, as the origin is an equilibrium for~\eqref{Delaysystem}, the positive orthant $\R^n_+$ is forward invariant, \textit{i.e.}, $x(t,\varphi)\in\R^n_+$ for all $t\in\R_+$ when $\varphi(t)\in\R^n_+$ for $t\in[-\tau_{\max},0]$.

We are interested in stability of the time-delay system (\ref{Delaysystem}) under the assumption that the delay-free system
\begin{align}
\dot{x}\bigl(t\bigr) = g\bigl(x(t),x(t)\bigr) \defr f\bigl(x(t)\bigr),
\label{eqn_sys_nondelayed}
\end{align}
has an asymptotically stable equilibrium point at the origin. Since time-delays may, in general, induce oscillations and even instability~\cite{Fridman:14}, the origin is not necessarily stable for the time-delay system (\ref{Delaysystem}). However, the following theorem shows that asymptotic stability of the origin for \textit{monotone} nonlinear systems is insensitive to time-delays satisfying Assumption~\ref{General Delays Assumption}.

\begin{theorem}
\label{thm_delayindependent}
\textit{
Consider the time-delay system~\eqref{Delaysystem} under Assumption~\ref{Assumption_delaysystem}. Then, the following statements are equivalent:
\begin{enumerate}
\item[1)] The time-delay monotone system~(\ref{Delaysystem}) has an asymptotically stable equilibrium point at the origin for all time varying-delays satisfying Assumption~\ref{General Delays Assumption}.
\item[2)] For the non-delayed monotone system~(\ref{eqn_sys_nondelayed}), any of the equivalent conditions in the statement of Theorem~\ref{thm_omegapath} hold.
\end{enumerate}
}
\end{theorem}

\begin{proof}
The proof is given in Appendix~B.
\end{proof}

According to Theorem~\ref{thm_delayindependent}, local asymptotic stability of the origin for a delay-free monotone system of the form~\eqref{eqn_sys_nondelayed} implies local asymptotic stability of the origin also for~(\ref{Delaysystem}) with bounded and unbounded time-varying delays. Theorem~\ref{thm_delayindependent} does not explicitly give any estimate of the region of attraction for the origin. However, its proof shows that the
stability conditions presented in Theorem~\ref{thm_omegapath} for non-delayed monotone systems can provide such estimates, leading to the following practical tests.
\begin{itemize}
\item[\textbf{T1.}] Assume that for the delay-free monotone system~(\ref{eqn_sys_nondelayed}), we can characterize asymptotic stability of the origin through a max-separable Lyapunov function $V$  satisfying the second statement of Theorem~\ref{thm_omegapath}. Then, for the time-delay system~(\ref{Delaysystem}), the origin is asymptotically stable with respect to initial conditions satisfying 
         \begin{align*}
    \varphi(t)\in \bigl\{x\in\R^n_+ \;|\; 0\leq x_i \leq V_i^{-1}(c),\; i\in\Ical_n\bigr\},
    \end{align*}
for $t\in[-\tau_{\max},0]$ with 
\begin{align*}
c=\min_{i \in \Ical_n} V_i(v_i).
\end{align*}
\item[\textbf{T2.}]  If we demonstrate the existence of a function $\rho$ such that the non-delayed system~(\ref{eqn_sys_nondelayed}) satisfies the third statement  of Theorem~\ref{thm_omegapath}, then~(\ref{Delaysystem}) with time-delays satisfying Assumption~\ref{General Delays Assumption} has an asymptotically stable equilibrium point at the origin for which the region of attraction includes initial conditions $\varphi$ that satisfy
    \begin{align*}
    0\leq \varphi(t)\leq \rho(\bar{s}),\quad t\in[-\tau_{\max},0].
    \end{align*}

\item[\textbf{T3.}] If we find a vector $w>0$ such that $g(w,w)<0$ and that the solution $x(t,w)$ to the delay-free monotone system~(\ref{eqn_sys_nondelayed}) converges to the origin, then the solution $x(t,\varphi)$ to the time-delay system~\eqref{Delaysystem} converges to the origin for any initial condition $\varphi$ that satisfies
    \begin{align*}
    0\leq \varphi(t)\leq w, \quad t\in[-\tau_{\max},0].
    \end{align*}
\end{itemize}

The following example illustrates the results of Theorem~\ref{thm_delayindependent}.

\begin{Example}
\label{DelayExample}
Consider the time-delay system
\begin{align}
\dot{x}(t)= g\bigl(x(t),x(t-\tau(t))\bigr)= \begin{bmatrix} -5x_1\bigl(t\bigr)+x_1\bigl(t\bigr)x_2^2\bigl(t-\tau(t)\bigr) \\  x_1\bigl(t-\tau(t)\bigr)- 2x_2^2\bigl(t\bigr) \end{bmatrix}.
\label{Example2}
\end{align}
One can verify that $g$ satisfies Assumption~\ref{Assumption_delaysystem}. Thus, the system~\eqref{Example2} is monotone on $\R^2_+$. According to Example~\ref{nondelayedExample}, this system without time-delays has an asymptotically stable equilibrium at the origin. Therefore, Theorem~\ref{thm_delayindependent} guarantees that for the time-delay system~\eqref{Example2}, the origin is still asymptotically stable for any bounded and unbounded time-varying delays satisfying Assumption~\ref{General Delays Assumption}.

We now provide an estimate of the region of attraction for the origin. Example~\ref{nondelayedExample} shows that for the system~\eqref{Example2} without time-delays, the function $\rho(s)=(s,\sqrt{s})$, $s\in[0,4]$, satisfies the third statement of Theorem~\ref{thm_omegapath}. It follows from stability test \textbf{T2} that the solution $x(t,\varphi)$ to~\eqref{Example2} starting from initial conditions
\begin{align*}
0\leq \varphi(t)\leq (\begin{array}{cc} 4, & 2 \end{array})^{\T},\quad t\in[-\tau_{\max},0].
\end{align*}
converges to the origin.
\end{Example}

\begin{remark}
Our results can be extended to monotone nonlinear systems with heterogeneous delays of the form
\begin{align}
\dot{x}_i(t) &= g_i\bigl(x(t),x^{\tau_i}(t)\bigr),\quad i\in\Ical_n,
\label{sys_heterogeneous}
\end{align}
where $g\bigl(x,y\bigr)=\bigl(g_1(x,y),\ldots,g_n(x,y)\bigr)^{\T}$ satisfies Assumption~\ref{Assumption_delaysystem}, and
\begin{align*}
x^{\tau_i}\bigl(t\bigr):=\bigl(x_1(t-\tau_i^1(t)),\ldots,x_n(t-\tau_i^n(t))\bigr)^{\T}.
\end{align*}
If the delays $\tau_i^j$, $i,j\in\Ical_n$, satisfy Assumption~\ref{General Delays Assumption}, then asymptotic stability of the origin for the delay-free monotone system~(\ref{eqn_sys_nondelayed}) ensures that~\eqref{sys_heterogeneous} with heterogeneous time-varying delays also has an asymptotically stable equilibrium point at the origin
\end{remark}

%
%

\section{Applications of the main results}
\label{sec_discussion}

In this section, we will present several examples to illustrate how our main results recover and generalize previous results on delay-independent stability of monotone nonlinear systems.

\subsection{Homogeneous monotone systems}
First, we consider a particular class of monotone nonlinear systems  whose vector fields are \textit{homogeneous} in the sense of the following definition.

\begin{definition}
Given an $n$-tuple $r = (r_1,\ldots,r_n)$ of positive real numbers and $\lambda> 0$,  the \textit{dilation map} $\delta^{r}_{\lambda}: \R^n \rightarrow \R^n$ is defined as
\begin{align*}
\delta^{r}_{\lambda}\bigl(x\bigr) := \bigl(\lambda^{r_1}x_1,\ldots,\lambda^{r_n}x_n\bigr).
\end{align*}
When $r =\ones_n$, the dilation map is called the standard dilation map. A vector field $f :\R^n \rightarrow \R^n$ is said to be \textit{homogeneous of degree $p\in\R$} with respect to the dilation map $\delta^{r}_{\lambda}$  if
\begin{align*}
f\bigl(\delta^{r}_{\lambda}(x)\bigr)=\lambda^{p}\delta^{r}_{\lambda}\bigl(f(x)\bigr),\quad \forall x\in\R^n,\; \forall \lambda>0.
\end{align*}
\vskip2mm
\end{definition}
Note that the linear mapping $f(x)=Ax$ is homogeneous of degree zero with respect to the standard dilation map.

The following result, which is a direct consequence of Theorem~\ref{thm_delayindependent}, establishes a necessary and sufficient condition for \textit{global} asymptotic stability of homogeneous monotone systems with time-varying delays. By global asymptotic stability, we mean that the origin is asymptotically stable for all nonnegative initial conditions.

\begin{Corollary}
\textit{
For the time-delay system~\eqref{Delaysystem}, suppose Assumption~\ref{Assumption_delaysystem} holds. Suppose also that $f(x):=g(x,x)$ is homogeneous of degree $p\in\R_+$ with respect to the dilation map $\delta^{r}_{\lambda}$. Then, the following statements are equivalent:
\begin{enumerate}
\item There exists a vector $w>0$ such that $f(w)<0$.
\item The homogeneous monotone system~\eqref{Delaysystem} has a globally asymptotically stable equilibrium point at the origin for all $\varphi\in \mathcal{C}\bigl([-\tau_{\max},0],\R_+^n\bigr)$ and all time-delays satisfying Assumption~\ref{General Delays Assumption}.
\end{enumerate}
}
\end{Corollary}

\begin{proof}
The implication $2)\Rightarrow 1)$ follows directly from Theorem~\ref{thm_delayindependent} and Remark~\ref{remark_equ}. We will show that $1)$ implies $2)$.

$1)\Rightarrow 2):$ Let $\rho_i(s)=s^{r_i/r_{\max}}w_i$, where
\begin{align*}
r_{\max}=\max_{i\in\Ical_n}\; r_i.
\end{align*}
For any positive constant $\bar{s}>0$, it is clear that $\rho_i$ is of class $\classK$ on $[0,\bar{s}]$, and $\smash{\rho_i^{-1}}$ is continuously differentiable and satisfy~\eqref{eqn_thm_omegapath_drhods}. As $f$ is homogeneous of degree $p\in\R_+$ with respect to the dilation map $\delta^{r}_{\lambda}$, it follows that
\begin{align}
f\circ\rho(s)=g\bigl(\rho(s),\rho(s)\bigr)&=g\bigl(\delta^{r}_{s^{1/r_{\max}}}(w),\delta^{r}_{s^{1/r_{\max}}}(w)\bigr)\nonumber\\
&= s^{p/r_{\max}}\delta^{r}_{s^{1/r_{\max}}}\bigl(g(w,w)\bigr)\nonumber\\
&=s^{p/r_{\max}}\delta^{r}_{s^{1/r_{\max}}}\bigl(f(w)\bigr).
\label{homogeneous_1}
\end{align}
Since $f(w)<0$, the right-hand side of equality~\eqref{homogeneous_1} is negative definite for all $s\in[0,\bar{s}]$. Therefore, according to Theorem~\ref{thm_delayindependent} and stability test \textbf{T2}, the time-delay system~\eqref{Delaysystem}  has an asymptotically stable equilibrium point at the origin for any $0\leq \varphi(t)\leq\rho(\bar{s})$, $t\in[-\tau_{\max},0]$.

To prove global asymptotic stability, suppose we are given a nonnegative initial condition $\varphi\in \mathcal{C}\bigl([-\tau_{\max},0],\R_+^n\bigr)$. Since $w_i>0$ for each $i \in \Ical_n$ and $\varphi$ is continuous (hence, bounded) on $[-\tau_{\max},0]$, there exists a sufficiently large $\bar{s}>0$ such that $\varphi(t)\leq \rho(\bar{s})$ for all $t\in[-\tau_{\max},0]$. It now follows immediately from the argument in the previous paragraph that the origin is asymptotically stable with respect to any nonnegative initial condition $\varphi$.
\end{proof}

\begin{remark}
 Delay-independent stability of homogeneous monotone systems with time-varying delays satisfying Assumption~\ref{General Delays Assumption} was previously considered in~\cite{feyzmahdavian_2014} by using a max-separable Lyapunov function with components
\begin{align*}
V_i(x_i)=\rho_i^{-1}(x_i)=\left(\frac{x_i}{w_i}\right)^{r_{\max}/r_i},\quad i\in \Ical_n.
\end{align*}
Note, however, that the proof of Theorem~\ref{thm_delayindependent} differs significantly from the analysis in~\cite{feyzmahdavian_2014}.  The main reason for this is that the homogeneity assumption, which plays a key role in the stability proof in~\cite{feyzmahdavian_2014}, is not satisfied for general monotone systems.
\end{remark}

\subsection{Sub-homogeneous monotone systems}

Another important class of monotone nonlinear systems are those with \textit{sub-homogeneous} vector fields:

\begin{definition}
A vector field $f:\R^{n}_+ \rightarrow \R^{n}$ is said to be sub-homogeneous of degree $p\in\R$ if
\begin{align*}
f(\lambda x)\leq \lambda^{p}f(x),\quad \forall x \in \R^{n}_+,\;\forall \lambda\geq 1.
\end{align*}
\vskip2mm
\end{definition}

Theorem~\ref{thm_delayindependent} allows us to show that \textit{global} asymptotic stability of the origin for monotone systems whose vector fields are sub-homogeneous is insensitive to bounded and unbounded time-varying delays.

\begin{Corollary}
\label{Corollary_subhomogeneous}
\textit{
Consider the time-delay system~\eqref{Delaysystem} under Assumption~\ref{Assumption_delaysystem}. Suppose also that $f(x):=g(x,x)$ is sub-homogeneous of degree $p\in\R_+$. If the origin for the delay-free monotone system~\eqref{eqn_sys_nondelayed} is globally asymptotically stable, then the sub-homogeneous monotone system~(\ref{Delaysystem}) has a globally asymptotically stable equilibrium at the origin for all time-varying delays satisfying Assumption~\ref{General Delays Assumption}.
}
\end{Corollary}

\begin{proof}
The origin for the sub-homogeneous monotone system~\eqref{eqn_sys_nondelayed} is globally asymptotically stable. Thus, for any constant $\alpha>0$,  there exists a vector $w>0$ such that $\alpha \ones_n \leq w$ and the solution $x(t,w)$ to the delay-free system~(\ref{eqn_sys_nondelayed}) converges to the origin~\cite[Theorem 4.1]{Vahid:11}. It follows from Theorem~\ref{thm_delayindependent} and stability test \textbf{T3} that the time-delay system~\eqref{Delaysystem} is asymptotically stable with respect to initial conditions satisfying $0\leq \varphi(t)\leq \alpha \ones_n$, $t\in[-\tau_{\max},0]$.

To complete the proof, let $\varphi\in \mathcal{C}\bigl([-\tau_{\max},0],\R_+^n\bigr)$ be an arbitrary initial condition. As $\alpha >0$ and $\varphi$ is continuous (hence, bounded) on $[-\tau_{\max},0]$, we can find $\alpha>0$ such that $\varphi(t)\leq \alpha \ones_n$ for $t\in [-\tau_{\max},0]$. This together with the above observations implies that the origin is asymptotically stable for all nonnegative initial conditions.
\end{proof}

\begin{remark}
 In~\cite{Feyzmahdavian14:MTNS}, it was shown that global asymptotic stability of the origin for sub-homogeneous monotone systems is independent of \textit{bounded} time-varying delays. In this work, we establish insensitivity of sub-homogeneous monotone systems to the general class of possibly \textit{unbounded} delays described by Assumption~\ref{General Delays Assumption}, which includes bounded delays as a special case.
\end{remark}

\subsection{Sub-homogeneous (non-monotone) positive systems}
Finally, motivated by results in~\cite{Vahid:14}, we consider the time-delay system
\begin{equation}
\dot{x}(t) ={g}(x(t),x(t-\tau(t)))=h\bigl(x(t)\bigr)+d\bigl(x(t-\tau(t))\bigr).
\label{Delaysystem-example}
\end{equation}
We assume that $h$ and $d$ satisfy Assumption~\ref{Assumption_application}.

\begin{assumption}
\label{Assumption_application}
The following properties hold:
\begin{enumerate}
\item For each $i\in\Ical_n$, $h_i(x)\geq 0$ for $x\in\R^n_+$ with $x_i=0$;
\item For all $x\in\R^n_+$, $d(x)\geq 0$;
\item Both $h$ and $d$ are sub-homogeneous of degree $p\in\R_+$;
\item For any $x\in\R^n_+\setminus\{0\}$, there is $i\in\Ical_n$ such that
\begin{align*}
\sup\bigl\{d_i(z')\;|\;0&\leq z' \leq x\bigr\}<-\sup\bigl\{h_i(z)\;|\;0\leq z \leq x,z_i=x_i\bigr\}.
\end{align*}
\end{enumerate}
\vskip2mm
\end{assumption}

Note that under Assumption~\ref{Assumption_application}, the time-delay system~\eqref{Delaysystem-example} is not necessarily
monotone. However, Assumptions~\ref{Assumption_application}.1 and~\ref{Assumption_application}.2 ensure the positivity of~\eqref{Delaysystem-example}~\cite[Theorem 5.2.1]{book_smith_1995}.

In~\cite{Vahid:14}, it was shown that  if Assumption~\ref{Assumption_application} holds, then the positive nonlinear system~\eqref{Delaysystem-example} with \textit{constant} delays $(\tau(t) = \tau_{\max}$, $t\in\R_+)$  has a globally asymptotically stable equilibrium at the origin for all $\tau_{\max}\in\R_+$.  Theorem~\ref{thm_delayindependent} helps us to extend the result in~\cite{Vahid:14} to \textit{time-varying} delays satisfying Assumption~\ref{General Delays Assumption}.

\begin{Corollary}
\textit{
For the time-delay system~\eqref{Delaysystem-example}, suppose Assumption~\ref{Assumption_application}  holds. Then, the origin is globally asymptotically stable for all time-delays satisfying Assumption~\ref{General Delays Assumption}.
}
\end{Corollary}

\begin{proof}
For any $x,y\in\R^n_+$ and each $i\in\Ical_n$, define
\begin{align*}
\bar{g}_i(&x,y)=\sup\bigl\{h_i(z)+d_i(z')\;|\;0\leq z \leq x,z_i=x_i,0\leq z' \leq y\bigr\}.
\end{align*}
It is straightforward to show that $\bar{g}(x,y)$ satisfies Assumption~\ref{Assumption_delaysystem}. Thus, the time-delay system
\begin{align}
\dot{x}\bigl(t\bigr) =\bar{g}\bigl(x(t),x(t-\tau(t))\bigr),
\label{eqn_sys_delayednew}
\end{align}
 is monotone. Under Assumption~\ref{Assumption_application}, the sub-homogeneous monotone system~\eqref{eqn_sys_delayednew} without delays $(\tau(t)=0)$ has a globally asymptotically stable equilibrium at the origin~\cite[Theorem III.2]{Vahid:14}. Therefore, according to Corollary~\ref{Corollary_subhomogeneous}, the origin for the time-delay system~\eqref{eqn_sys_delayednew} is also globally asymptotically stable for any time-delays satisfying Assumption~\ref{General Delays Assumption}.

As $g(x,y)\leq \bar{g}(x,y)$ for any $x,y\in\R^n_+$, it follows from~\cite[Theorem 5.1.1]{book_smith_1995} that for any initial condition $\varphi$,
\begin{align}
x(t,\varphi,g)\leq x(t,\varphi,\bar{g}),\quad t\in\R_+,
\label{order_nonmonotone}
\end{align}
where $x(t,\varphi,g)$ and $ x(t,\varphi,\bar{g})$ are solutions to~\eqref{Delaysystem-example} and~\eqref{eqn_sys_delayednew}, respectively, for a common initial condition $\varphi$. Since $x=0$ is a globally asymptotically stable equilibrium point for~\eqref{eqn_sys_delayednew}, $x(t,\varphi,\bar{g})\rightarrow 0$ as $t \rightarrow \infty$. Moreover, as~\eqref{Delaysystem-example} is a positive system, $x(t,\varphi,g)\geq 0$ for $t\in\R_+$. We can conclude from~\eqref{order_nonmonotone} and the above observations that for any nonnegative initial condition $\varphi$, $x(t,\varphi,g)$ converges to the origin. Hence, for the time-delay system~\eqref{Delaysystem-example}, the origin is globally asymptotically stable.
\end{proof}

%
%

\section{Conclusions}
\label{sec_conclusions}

In this paper, we have presented a number of results that extend fundamental stability properties of monotone linear systems to monotone nonlinear systems. Specifically, we have shown that for such nonlinear systems, equivalence holds between asymptotic stability of the origin, the existence of a max-separable Lyapunov function on a compact set, and D-stability. In addition, we have demonstrated that if the origin for a delay-free monotone system is asymptotically stable, then the corresponding system with bounded and unbounded time-varying delays also has an asymptotically stable equilibrium point at the origin. We have derived a set of necessary and sufficient conditions for establishing delay-independent stability of monotone nonlinear systems, which allow us to extend several earlier works in the literature. We have illustrated the main results with several examples.

%
%

\appendix

Before proving the main results of the paper, namely, Theorems~\ref{thm_omegapath} and~\ref{thm_delayindependent}, we first state a key lemma which shows that all components of a max-separable Lyapunov function are necessarily monotonically increasing.

\begin{Lemma}
\label{lem_Vipositivederivative}
\textit{
Consider the monotone system (\ref{eqn_sys}) and the max-separable function $V:\R_+^n\rightarrow\R_+$ as in (\ref{eqn_maxsepV}) with $V_i:\R_+\rightarrow\R_+$ differentiable for all $i\in\Ical_n$. Suppose that there exist functions $\nu_1,\nu_2$ of class \classK such that
\begin{align}
\nu_1(x_i) \leq V_i(x_i) \leq \nu_2(x_i), \label{eqn_lem_Viposder_Vibounds}
\end{align}
for all $x_i\in\R_+$ and all $i\in\Ical_n$. Suppose also that there exists a positive definite function $\mu$ such that
\begin{align}
D^+V(x) \leq -\mu(V(x)), \label{eqn_lem_Viposder_DV}
\end{align}
for all $x\in\R_+^n$. Then, the functions $V_i$ satisfy, for all $x_i >0$,}
\begin{align}
\pd{V_i}{x_i}(x_i) > 0.
\label{eqn_lem_Viposder_Viposder}
\end{align}
\vskip2mm%
\end{Lemma}

\begin{proof}
For some $j\in\Ical_n$, consider the state $x = e_jx_j$, where $e_j$ is the $j^{\textup{th}}$ column of the identity matrix $I\in\R^{n\times n}$, and $x_j\in \R_+$. From (\ref{eqn_lem_Viposder_Vibounds}), $V_j(x_j)>0$ for any $x_j>0$ and $V_i(x_i)=0$ for $i \neq j$. Thus, the set $\Jcal$ in (\ref{eqn_lem_DVmaxsep_Jcal}) satisfies $\Jcal(e_jx_j) = \{j\}$ for all $x_j>0$. Evaluating (\ref{eqn_lem_Viposder_DV}) through Proposition~\ref{lem_DVmaxsep} leads to
\begin{align}
\!\!D^+V(e_jx_j) = \pd{V_j}{x_j}(x_j)f_j(e_jx_j) \leq -\mu(V(e_jx_j)) < 0\!
\label{eqn_lem_Viposder_proof_step1}
\end{align}
for $x_j>0$. The strict inequality in (\ref{eqn_lem_Viposder_proof_step1}) implies that ${\partial V_j}/{\partial x_j}$ is nonzero and that its sign is constant for all $x_j>0$. As a negative sign would yield $V_j(x_j)<0$ (and, hence violate (\ref{eqn_lem_Viposder_Vibounds})) and $j$ is chosen arbitrarily, (\ref{eqn_lem_Viposder_Viposder}) holds.
\end{proof}

\subsection*{A.~Proof of Theorem~\ref{thm_omegapath}}

The theorem will be proven by first showing the equivalence $1) \Leftrightarrow 2) \Leftrightarrow 3)$. Next, this equivalence will be exploited to subsequently show $3)\Rightarrow4)$ and $4)\Rightarrow1)$. Consequently, we have
\begin{align*}
1) \Leftrightarrow 2) \Leftrightarrow 3) \Rightarrow 4) \Rightarrow 1),
\end{align*}
which proves the desired result.

First, we note that the implication $2)\Rightarrow1)$ follows directly from Lyapunov stability theory (\textit{e.g.}, \cite{book_rouche_1977}) and we proceed to prove that $1)$ implies $2)$.

$1)\Rightarrow 2)$: By asymptotic stability of the origin as in Definition~\ref{def_stability}, the region of attraction of $x=0$ defined as
\begin{align*}
\textstyle
\Acal := \bigl\{ x_0\in\R^n_+ \mid \lim_{t\rightarrow\infty}x(t,x_0) = 0 \bigr\}
\end{align*}
is nonempty. In fact, we can find some $\delta>0$ such that all states satisfying $0\leq x< \delta\ones$ are in $\Acal$. Then, as the system~(\ref{eqn_sys}) is monotone, the reasoning from~\cite[Theorem~3.12]{ruffer_2010b} (see also \cite[Theorem~2.5]{ruffer_2010c} for a more explicit statement) can be followed to show the existence of a vector $v$  such that $0< v< \delta\ones$ and $f(v)< 0$.

Let $\omega(t)=x(t,v)$, $t\in \R_+$ be the solution to~\eqref{eqn_sys} starting from such a $v$. By the local Lipschitz continuity of $f$, $\omega$ is continuously differentiable. As $v\in {\mathcal A}$,  $w_i(t)\rightarrow 0$ when $t\rightarrow \infty$ for each $i \in \Ical_n$. Moreover, note that $v$ is an element of the set
\begin{align*}
\Omega = \{x\in\R^n_+\mid f(x) <0\}.
\end{align*}
According to~\cite[Proposition 3.2.1]{book_smith_1995}, $\Omega$ is forward invariant so $\omega(t) \in\Omega$ for all $t\in\R_+$. Thus, the components $\omega_i(t)$ are strictly decreasing in $t$, \textit{i.e.}, $\dot{\omega}_i(t)<0$ for $t\in \R_+$. This further implies that, for a given state component $x_i\in (0,v_i]$, there exists a unique $t\in\R_+$ such that $x_i=\omega_i(t)$. Let
\begin{align*}
T_i(x_i)=\bigl\{t\in \R_+ \;|\; x_i=\omega_i(t)\bigr\}.
\end{align*}
From the definition, it is clear that $T_i(x_i)=\omega_i^{-1}(x_i)$. Since $\dot{\omega}_i(t)<0$ for $t\in \R_+$, the inverse of $\omega_i$, \textit{i.e.}, the function $T_i$, is continuously differentiable and strictly decreasing for all $x_i\in (0,v_i]$. We define, as in \cite{rantzer_2013,dirr_2015}, the component functions
\begin{align}
V_i(x_i)=e^{-T_i(x_i)},\quad i\in\Ical_n.
\label{proof12_Vi}
\end{align}
Note that $V_i(0)=0$. Moreover, $V_i$ are continuously differentiable and strictly increasing for all $x_i\in (0,v_i]$ as a result of the properties of the function $T_i$. Therefore, the component functions $V_i$ in (\ref{proof12_Vi}) satisfy (\ref{eqn_thm_omegapath_Vbounds}) for some functions $\nu_1,\nu_2$ of class \classK. Moreover, from \cite[Theorem 3.2]{dirr_2015}, the upper-right Dini derivative of the max-separable Lyapunov function (\ref{eqn_maxsepV}) with components (\ref{proof12_Vi}) is given by
\begin{align*}
D^+V(x) \leq - V(x),
\end{align*}
for all $x\in\Xcal$ with $\Xcal = \{x\in\R^n_+\mid 0\leq x\leq v\}$. This shows that (\ref{eqn_thm_omegapath_DV}) holds, and hence the proof is complete.$\qed$

We now proceed to show equivalence between $2)$ and $3)$. We begin with the implication $2)\Rightarrow 3)$.

$2)\Rightarrow 3)$:  According to Lemma~\ref{lem_Vipositivederivative}, the functions $V_i$ are monotonically increasing. Thus, their inverses
\begin{align}
\rho_i(s) = V_i^{-1}(s)
\label{eqn_thm_omegapath_proof12_rhoi}
\end{align}
can be defined for all $s\in[0,V_i(v_i)]$. Define
\begin{align*}
\bar{s}:=\min_{i \in \Ical_n} V_i(v_i).
\end{align*}
Since $v>0$, it follows from~\eqref{eqn_thm_omegapath_Vbounds} that $\bar{s}>0$. Moreover, due to continuous differentiability of $V_i$ and Lemma~\ref{lem_Vipositivederivative},  $\rho_i$ is of class $\classK$ and satisfies (\ref{eqn_thm_omegapath_drhods}).

In the remainder of the proof, it will be shown that the function $\rho$ with components $\rho_i$ defined in (\ref{eqn_thm_omegapath_proof12_rhoi}) satisfies (\ref{eqn_thm_omegapath_rhoineq}) for all $s\in[0,\bar{s}]$. Thereto, consider a state ${x}^s$, parameterized by $s$ according to the definition ${x}^s := \rho(s)$.  Since $\bar{s} \leq V_i(v_i)$ for all $i \in \Ical_n$, it holds that ${x}^s\in{\Xcal}$ for all $s\in[0,\bar{s}]$. Evaluating (\ref{eqn_thm_omegapath_DV}) for such an $x^s$ yields
\begin{align}
\!D^+V\bigl({x}^s\bigr) = \max_{j\in\Jcal({x}^s)} \pd{V_j}{x_j}\bigl({x}^{s}_{j}\bigr)f_j\bigl({x}^s\bigr) \leq -\mu\bigl(V({x}^s)\bigr).
\label{eqn_thm_omegapath_proof12_DV_step1}
\end{align}
The definition ${x}_i^s = \rho_i(s)$, $i\in\Ical_n$, implies, through (\ref{eqn_thm_omegapath_proof12_rhoi}), that $V_i({x}^{s}_{i}) = s$ for all $s\in[0,\bar{s}]$. Consequently, $V({x}^s) = s$ and the set $\Jcal({x}^s)$ in (\ref{eqn_thm_omegapath_proof12_DV_step1}) satisfies $\Jcal({x}^s) = \Ical_n$, such that (\ref{eqn_thm_omegapath_proof12_DV_step1}) implies
\begin{align}
\pd{V_i}{x_i}\bigl(\rho_i(s)\bigr)f_i\bigl(\rho(s)\bigr) \leq -\mu\bigl(s\bigr)
\label{Proof23_step1}
\end{align}
for $s\in[0,\bar{s}]$ and $i\in\Ical_n$. Since $V_i$ is strictly increasing and $\mu$ is positive definite, $f_i(\rho_i(s))\leq0$. Define functions $r_i:[0,\bar{s}]\rightarrow\R$ as
\begin{align*}
r_i(s) \defl \sup\!\bigg\{ \pd{V_i}{x_i}(z_i) \;\bigg|\;  \rho_i(s) \leq z_i \leq \rho_i(\bar{s}) \bigg\}.
\end{align*}
By continuous differentiability of $V_i$ and the result of Lemma~\ref{lem_Vipositivederivative}, it follows that $r_i$ exists and satisfies $r_i(s)>0$ for all $s\in[0,\bar{s}]$. Moreover, it is easily seen that
\begin{align*}
r_i(s) \geq \pd{V_i}{x_i}(\rho_i(s)).
\end{align*}
This together with~\eqref{Proof23_step1} implies that
\begin{align*}
r_i(s)f_i(\rho(s)) \leq \pd{V_i}{x_i}(\rho_i(s))f_i(\rho(s)) \leq -\mu(s)
\end{align*}
for all $s\in[0,\bar{s}]$. Here, the inequality follows from the observation that $f_i(\rho_i(s))\leq0$. Then, strict positivity of $r_i$ implies that
\begin{align*}
f_i(\rho(s)) \leq -\frac{\mu(s)}{r_i(s)},
\end{align*}
for all $s\in[0,\bar{s}]$ and any $i\in\Ical_n$. Since $r_i$ is strictly positive and $\mu$ is positive definite, the function $\alpha_i(s) = \mu(s)/r_i(s)$ is positive definite and (\ref{eqn_thm_omegapath_rhoineq}) holds. $\qed$

We continue with the reverse implication $3) \Rightarrow 2)$.

$3)\Rightarrow 2):$ Define $v \defl \rho(\bar{s})$. Since $\rho_i$, $i \in \Ical_n$, are of class $\classK$ and $\bar{s}>0$, we have $v>0$. Let $V_i$ be such that
\begin{align}
V_i(x_i) = \rho_i^{-1}(x_i)
\label{eqn_thm_omegapath_proof21_Vi}
\end{align}
for $x_i\in [0, v_i]$. Note that the inverse of $\rho_i$ exists on this compact set as $\rho_i$ is of class $\classK$ and, hence, is strictly increasing. Because of the same reason, it is clear that $V_i$ as in (\ref{eqn_thm_omegapath_proof21_Vi}) satisfies (\ref{eqn_thm_omegapath_Vbounds}) for some functions $\nu_1,\nu_2$ of class \classK.

The remainder of the proof will show that the max-separable Lyapunov function (\ref{eqn_maxsepV}) with components (\ref{eqn_thm_omegapath_proof21_Vi}) satisfies (\ref{eqn_thm_omegapath_DV}). By Proposition~\ref{lem_DVmaxsep}, the upper-right Dini derivative of $V$ along solutions of (\ref{eqn_sys}) is given by
\begin{align}
D^+V(x) = \max_{j\in\Jcal(x)} \pd{V_j}{x_j}(x_j)f_j(x).
\label{eqn_thm_omegapath_proof21_DV_step1}
\end{align}
Define $\Xcal\defl\{x\in\R_+^n \;|\; 0\leq x \leq v\}$, choose any $x\in\Xcal$ and consider any $j\in\Jcal(x)$, where $\Jcal(x)$ is defined in (\ref{eqn_lem_DVmaxsep_Jcal}). Then, $V_j(x_j) = V(x)$, such that the use of equality (\ref{eqn_thm_omegapath_proof21_Vi}) leads to
\begin{align}
x_j = \rho_j(V_j(x_j)) = \rho_j(V(x)).
\label{eqn_thm_omegapath_proof21_xjeq}
\end{align}
Note that when $i$ be different from $j$, a similar argument establishes that $V_i(x_i)\leq V(x)$ and, thus, $x_i\leq \rho_i(V(x))$. Combining this with (\ref{eqn_thm_omegapath_proof21_xjeq}) gives
\begin{align}
x \leq \rho(V(x)).
\label{eqn_thm_omegapath_proof21_xineq}
\end{align}
Since $f$ satisfies Kamke condition~(\ref{eqn_thm_kamkecondition}), it follows from of \eqref{eqn_thm_omegapath_proof21_xjeq} and \eqref{eqn_thm_omegapath_proof21_xineq} that $f_j(x)\leq f_j(\rho(V(x)))$. Moreover, ${\partial V_j}/{\partial x_j} > 0$ for all $x_j\in (0, v_j]$ due to~(\ref{eqn_thm_omegapath_proof21_Vi}) and the fact that $\rho_i$ satisfies (\ref{eqn_thm_omegapath_drhods}). The above observations together with~\eqref{eqn_thm_omegapath_proof21_DV_step1} imply that
\begin{align}
D^+V(x) \leq \max_{j\in\Jcal(x)} \pd{V_j}{x_j}(x_j)f_j\big(\rho(V(x))\big).
\label{eqn_thm_omegapath_proof21_DV_step2}
\end{align}
Next, define functions $\bar{r}_i:[0,\bar{s}]\rightarrow\R_+$ as
\begin{align}
\bar{r}_i(s) \defl \inf\!\bigg\{ \pd{V_i}{x_i}(z_i) \;\bigg|\; \rho_i(s) \leq z_i \leq v_i  \bigg\}.
\label{eqn_thm_omegapath_proof21_ri}
\end{align}
Note that $\bar{r}_i$ is strictly positive for $s\in(0,\bar{s}]$. Moreover, for any $j\in\Jcal(x)$, the equality (\ref{eqn_thm_omegapath_proof21_xjeq}) is recalled, from which it follows that
\begin{align}
 \bar{r}_j(V(x))=\bar{r}_j(V_j(x_j))&=  \inf\!\bigg\{ \pd{V_j}{x_j}(z_j) \;\bigg|\; x_j \leq z_j \leq v_i  \bigg\}\nonumber\\
&\leq \pd{V_j}{x_j}(x_j).
\label{eqn_thm_omegapath_proof21_rjbound}
\end{align}
Then, returning to the Dini derivative of $V$ in (\ref{eqn_thm_omegapath_proof21_DV_step2}), the use of (\ref{eqn_thm_omegapath_rhoineq}), and subsequently, the application of (\ref{eqn_thm_omegapath_proof21_rjbound}), leads to
\begin{align}
D^+V(x) &\leq \max_{j\in\Jcal(x)} -\pd{V_j}{x_j}(x_j)\alpha_j(V(x)), \nonumber\\
&\leq \max_{j\in\Jcal(x)} -\bar{r}_j(V(x))\alpha_j(V(x)),
\label{eqn_thm_omegapath_proof21_DV_step3}
\end{align}
for all $x\in\Xcal$. Recall that the functions $\alpha_i$ are positive definite, whereas the functions $\bar{r}_i$ in (\ref{eqn_thm_omegapath_proof21_ri}) are strictly positive for $s\in(0,\bar{s}]$. As a result, the functions $\bar{r}_i(s)\alpha_i(s)$ are positive definite and there exists a positive definite function $\mu$ such that $\mu(s) \leq \bar{r}_i(s)\alpha_i(s)$ for all $s\in[0,\bar{s}]$ and all $i\in\Ical_n$. Applying this result to (\ref{eqn_thm_omegapath_proof21_DV_step3}) yields
\begin{align*}
D^+V(x) \leq -\mu(V(x)),
\end{align*}
for all $x\in\Xcal$, proving (\ref{eqn_thm_omegapath_DV}) and finalizing the proof.$\qed$

We now show that $3)$ implies $4)$ by exploiting the equivalence $1)\Leftrightarrow2)\Leftrightarrow3)$.

$3)\Rightarrow 4)$: First, we show that the system (\ref{eqn_thm_dstab_sysphi}) is monotone. Thereto, recall that monotonicity of (\ref{eqn_sys}) implies
\begin{eqnarray}
\hspace{-4mm} x' \leq x,\; x'_i = x_i \;\;&\Rightarrow&\;\; f_i(x')\leq f_i(x) \nonumber\\
\;\;&\Rightarrow&\;\; \psi_i\bigl(x'_i,f_i(x')\bigr) \leq\psi_i\bigl(x_i,f_i(x)\bigr),
\label{eqn_thm_dstab_proof_step1}
\end{eqnarray}
where the latter implication follows from (\ref{eqn_thm_dstab_phimonotone}). Then, (\ref{eqn_thm_dstab_proof_step1}) represents the Kamke condition for the vector field $\psi(x,f)$, such that monotonicity of (\ref{eqn_thm_dstab_sysphi}) follows from Proposition~\ref{thm_kamkecondition}.

Next, note that $\varphi\bigl(0,f(0)\bigr)=0$ implying that the origin is an equilibrium point of (\ref{eqn_thm_dstab_sysphi}). In order to prove asymptotic stability of the origin, we recall that $f$ satisfies, by assumption, (\ref{eqn_thm_omegapath_rhoineq}) for some function $\rho$. From this, we have
\begin{align}
\psi\bigl(\rho(s),f(\rho(s)\bigr) \leq \psi\bigl(\rho(s),-\alpha(s)\bigr),
\label{eqn_thm_dstab_proof_step2}
\end{align}
where the inequality is maintained due to the fact that $\psi_i$ is monotonically increasing in the second argument for all $i\in\Ical_n$. The functions $\alpha_i$ and $\rho$ are positive definite, hence $-\alpha(s)<0$ and $\rho(s)>0$ for all $s\in(0,\bar{s}]$. Then, from~(\ref{eqn_thm_dstab_phimonotone}), we have
\begin{align*}
\psi\bigl(\rho(s),-\alpha(s)\bigr)<\psi\bigl(\rho(s),0\bigr)=0,
\end{align*}
for $s\in(0,\bar{s}]$, where the equality follows from $\varphi(x,0)=0$ for any $x\in\R_+^n$. Hence, $\psi\bigl(\rho(s),-\alpha(s)\bigr)$ is negative definite and (\ref{eqn_thm_dstab_proof_step2}) is again of the form (\ref{eqn_thm_omegapath_rhoineq}). From the implication $3)\Rightarrow 2)\Rightarrow 1)$, we conclude the origin is asymptotically stable.$\qed$

Finally, we prove that $4)$ implies $1)$.

$4)\Rightarrow 1)$: Assume that the system  (\ref{eqn_thm_dstab_sysphi}) is asymptotically stable for any Lipschitz continuous function $\psi$ satisfying statement $4)$. Particularly, let $\psi(x,y)=y$. Then, the monotone system~\eqref{eqn_sys} is asymptotically stable.$\qed$

\subsection*{B.~Proof of Theorem~\ref{thm_delayindependent}}

$1)\Rightarrow 2)$: Assume that $x=0$ for the time-delay system~\eqref{Delaysystem} is asymptotically stable for all delays satisfying Assumption~\ref{General Delays Assumption}. Particularly, let $\tau(t)= 0$. Then, the non-delayed monotone system~\eqref{eqn_sys_nondelayed} has an asymptotically stable equilibrium point at the origin.

$2)\Rightarrow 1)$: For investigating asymptotic  stability of the time-delay monotone system~\eqref{Delaysystem}, we employ Lyapunov-Razumikhin approach which allows us to impose minimal restrictions on time-varying delays~\cite{driver_1962}. In particular, we make use of the max-separable Lyapunov function that guarantees asymptotic stability of the origin without time-delays as a Lyapunov-Razumikhin function.

Let $\rho$ be a function such that the delay-free system~\eqref{eqn_sys_nondelayed} satisfies (\ref{eqn_thm_omegapath_rhoineq}), \textit{i.e.},
\begin{align}
        g\bigl(\rho(s),\rho(s)\bigr) \leq -\alpha\bigl(s\bigr),
        \label{prooftheorem2-1}
\end{align}
holds for $s\in [0,\bar{s}]$. Define $v:=\rho(\bar{s})$ and let $\Xcal$ be the compact set (\ref{eqn_thm_omegapath_Xcal}). First, we show that for any $\varphi(t)\in \Xcal$, $t\in[-\tau_{\max},0]$, the solution $x(t,\varphi)$ satisfies $x(t,\varphi)\in \Xcal$ for all $t\in\R_+$.

Clearly, $x(0,\varphi)=\varphi(0)\in \Xcal$. In order to establish a contradiction, suppose that the statement $x(t,\varphi)\in \Xcal$, $t\in\R_+$, is not true. Then, there is $i\in\Ical_n$ and a time $\hat{t}\in\R_+$ such that $x(t,\varphi)\in \Xcal$ for all $t\in[0,\hat{t}]$, $x_i(\hat{t},\varphi)=v_i$, and
\begin{equation}
\label{prooftheorem2-2}
D^+x_i(\hat{t},\varphi)\geq 0.
\end{equation}
As $x(\hat{t},\varphi)\leq v$, it follows from Assumption~\ref{Assumption_delaysystem}.1 that
\begin{align}
g_i\bigl(x(\hat{t},\varphi),y\bigr)\leq g_i\bigl(v,y\bigr),
\label{prooftheorem2-3}
\end{align}
for $y\in\R_+^n$. As $\hat{t}-\tau(\hat{t})\in [-\tau_{\max},\hat{t}]$ and $\varphi(t)\in \Xcal$ for $t\in[-\tau_{\max},0]$, we have $x(\hat{t}-\tau(\hat{t}),\varphi) \in \Xcal$ irrespectively of whether $\hat{t}-\tau(\hat{t})$ is nonnegative or not. Thus, from Assumption~\ref{Assumption_delaysystem}.2,
\begin{align}
g_i\bigl(y,x(\hat{t}-\tau(\hat{t}),\varphi)\bigr)\leq g_i\bigl(y,v\bigr),
\label{prooftheorem2-4}
\end{align}
for any $y\in\R_+^n$. Using~\eqref{prooftheorem2-3} and~\eqref{prooftheorem2-4}, the Dini-derivative of $x_i(t,\varphi)$ along the trajectories of~\eqref{Delaysystem} at $t=\hat{t}$ is given by
\begin{align*}
D^+x_i(\hat{t},\varphi)&=g_i\bigl(x(\hat{t},\varphi),x(\hat{t}-\tau(\hat{t}),\varphi)\bigr)\\
&\leq g_i(v,v)= g_i\bigl(\rho(\bar{s}),\rho(\bar{s})\bigr)\\
&\leq -\alpha_i(\bar{s})<0,
\end{align*}
which contradicts~\eqref{prooftheorem2-2}. Therefore, $x(t,\varphi)\in \Xcal$ for $t\in\R_+$.

We now prove the asymptotic stability of the origin. According to the proof of Theorem~\ref{thm_omegapath}, if the delay-free system~(\ref{eqn_sys_nondelayed}) satisfies~(\ref{prooftheorem2-1}), then it admits a max-separable Lyapunov function (\ref{eqn_maxsepV}) with components $V_i(x_i) = \rho_i^{-1}(x_i)$ defined on $\Xcal$, see~(\ref{eqn_thm_omegapath_proof21_Vi}), such that
\begin{align}
D^+V(x)=\max_{j\in\Jcal({x})} \pd{V_j}{x_j}\bigl({x}_{j}\bigr)g_j\bigl(x,{x}\bigr)  \leq -\mu(V(x)),
\label{prooftheorem2-5}
\end{align}
holds for all $x\in \Xcal$. For any $\varphi(t)\in {\Xcal}$, from~\eqref{eqn_thm_omegapath_proof21_Vi}, we have
\begin{align}
x_j\bigl(t,\varphi\bigr) = \rho_j\bigl( V(x(t,\varphi)) \bigr), \quad j\in\Jcal(x(t,\varphi)),
\label{eqn_thm_delayindependent_proof_step1}
\end{align}
with $\Jcal$ as in (\ref{eqn_lem_DVmaxsep_Jcal}). Combining (\ref{eqn_thm_delayindependent_proof_step1}) with the observation that
\begin{align*}
x_i\bigl(t,\varphi\bigr) \leq \rho_i\bigl( V(x(t,\varphi)) \bigr),
\end{align*}
for any $i\in\Ical_n$, implying that
\begin{align}
x(t,\varphi) \leq \rho\bigl( V(x(t,\varphi)) \bigr) \defr \bar{x}(t)
\label{eqn_thm_delayindependent_proof_xineq}
\end{align}
for $\varphi(t)\in {\Xcal}$. At this point we recall that $x(t,\varphi)\in\Xcal$ for all $t\in\R_+$. Thus, $V(x(t,\varphi))\leq V(v)$. This in turn implies that $\rho(V(x(t,\varphi)) \bigr)\leq v$ and, hence, $\bar{x}(t)\in{\Xcal}$ for any $t\in \R_+$. Next, it follows from Assumption~\ref{Assumption_delaysystem}.1 that
\begin{align}
g_j\bigl(x(t,\varphi),y\bigr)\leq g_j\bigl(\bar{x}(t),y\bigr),\quad  j\in\Jcal(x(t,\varphi)),
\label{eqn_thm_delayindependent_proof_step2}
\end{align}
for any $\varphi(t)\in {\Xcal}$ and any $y\in\R_+^n$. This condition will be exploited later in the proof.

In the remainder of the proof, the max-separable Lyapunov function $V$ of the delay-free system (\ref{eqn_sys_nondelayed}) will be used as a candidate Lyapunov-Razumikhin function for the time-delay system (\ref{Delaysystem}). To establish a Razumikhin-type argument~\cite{driver_1962}, it is assumed that
\begin{align}
V\bigl(x(s,\varphi)\bigr) < q\bigl(V(x(t,\varphi))\bigr),
\label{eqn_thm_delayindependent_proof_razumikhincondition}
\end{align}
for all $s\in[t-\tau(t),t]$, where $q:\R_+\rightarrow \R_+$ is a continuous non-decreasing function satisfying $q(r)>r$ for all $r>0$. We will specify $q$ later. By the definition of $V$, it follows that the assumption (\ref{eqn_thm_delayindependent_proof_razumikhincondition}) implies that
\begin{align*}
V_i(x_i(t-\tau(t),\varphi))\leq q(V(x(t,\varphi))),
\end{align*}
for all $i\in\Ical_n$. As a result,
\begin{align}
x\bigl(t-\tau(t),\varphi\bigr) \leq \rho\bigl(q(V(x(t,\varphi))) \bigr) =: \tilde{x}(t),
\label{eqn_thm_delayindependent_proof_xdelayineq}
\end{align}
such that the application of Assumption~\ref{Assumption_delaysystem}.2 yields
\begin{align}
g_i\bigl(y,x(t-\tau,\varphi)\bigr)\leq g_i\bigl(y,\tilde{x}(t)\bigr),
\label{eqn_thm_delayindependent_proof_step3}
\end{align}
for all $i\in\Ical_n$ and any $y\in\R_n^+$. Returning to the candidate Lyapunov-Razumikin function $V$, its upper-right Dini derivative along trajectories $x(\cdot,\varphi)$ reads
\begin{align}
D^+V\bigl(x(t,\varphi)\bigr)&= \max_{j\in\Jcal(x(t,\varphi))} \pd{V_j}{x_j}\bigl(x_j(t,\varphi)\bigr)g_j\bigl(x(t,\varphi),x(t-\tau(t),\varphi)\bigr)\nonumber\\
&\leq  \max_{j\in\Jcal(x(t,\varphi))} \pd{V_j}{x_j}\bigl(x_j(t,\varphi)\bigr)g_j\bigl(\bar{x}(t),x(t-\tau(t),\varphi)\bigr)\nonumber\\
&= \max_{j\in\Jcal(x(t,\varphi))} \pd{V_j}{x_j}\bigl(\bar{x}_j(t)\bigr)g_j\bigl(\bar{x}(t),x(t-\tau(t),\varphi)\bigr),
\label{eqn_thm_delayindependent_proof_DV_step1}
\end{align}
where (\ref{eqn_thm_delayindependent_proof_step2}) was used to obtain the inequality (exploiting Proposition~\ref{lem_Vipositivederivative} as before) and~\eqref{eqn_thm_delayindependent_proof_step1} to get the second equality. Next, recall that the assumption (\ref{eqn_thm_delayindependent_proof_razumikhincondition}) implies (\ref{eqn_thm_delayindependent_proof_step3}), such that (\ref{eqn_thm_delayindependent_proof_DV_step1}) is bounded as
\begin{align}
\hspace{-3mm}D^+V\bigl(x(t,\varphi)\bigr) &\leq \max_{j\in\Jcal(x(t,\varphi))} \pd{V_j}{x_j}\bigl(\bar{x}_j(t)\bigr)g_j\bigl(\bar{x}(t),\tilde{x}(t)\bigr)\nonumber\\
&\leq \max_{j\in\Jcal(\bar{x}(t))} \pd{V_j}{x_j}\bigl(\bar{x}_j(t)\bigr)g_j\bigl(\bar{x}(t),\tilde{x}(t)\bigr).
\label{eqn_thm_delayindependent_proof_DV_step2}
\end{align}
Here, the second inequality follows from the observation that $\Jcal(\bar{x}(t)) = \Ical_n$ for any $t$, such that $\Jcal(x(t,\varphi))\subseteq\Jcal(\bar{x}(t))$.

We recall that the value of $\tilde{x}(t)$ in (\ref{eqn_thm_delayindependent_proof_DV_step2}) is dependent on the choice of the function $q$, see (\ref{eqn_thm_delayindependent_proof_xdelayineq}). At this point, we assume that $q$ can be chosen such that
\begin{align}
0\leq g_i\bigl(\bar{x}(t),\tilde{x}(t)\bigr) - g_i\bigl(\bar{x}(t),\bar{x}(t)\bigr) \leq \frac{\mu(V(\bar{x}(t)))}{2kD},
\label{eqn_thm_delayindependent_proof_q_assumption}
\end{align}
holds for all $\bar{x}(t)\in\Xcal$ and for any $i\in\Ical_n$. Here, $k\geq1$ will be chosen later and $D>0$ is such that
\begin{align}
\pd{V_i}{x_i}(x_i) \leq D, \quad \forall x\in\Xcal,\; \forall i\in\Ical_n.
\label{eqn_thm_delayindependent_proof_dVi_bound}
\end{align}
Note that such $D$ exists by continuous differentiability of $V$ and the fact that $\Xcal$ is a compact set. Also, we stress that the first inequality in (\ref{eqn_thm_delayindependent_proof_q_assumption}) follows from the observation that $\bar{x}(t)\leq\tilde{x}(t)$ (compare (\ref{eqn_thm_delayindependent_proof_xineq}) with (\ref{eqn_thm_delayindependent_proof_xdelayineq}) and recall that $q(r) > r$ for all $r>0$ and the functions $\rho_i$ are of class \classK).

Now, under the assumption (\ref{eqn_thm_delayindependent_proof_q_assumption}), it follows from (\ref{eqn_thm_delayindependent_proof_DV_step2}) that
\begin{align}
D^+V\bigl(x(t,\varphi)\bigr) &\leq \max_{j\in\Jcal(\bar{x}(t))}\biggl\{ \pd{V_j}{x_j}\bigl(\bar{x}_j(t)\bigr)g_j\bigl(\bar{x}(t),\bar{x}(t)\bigr)+ \pd{V_j}{x_j}\bigl(\bar{x}_j(t)\bigr)\frac{{\mu}\bigl(V(\bar{x}(t))\bigr)}{2kD }\biggr\}\nonumber\\
&\leq -\left(1-\frac{1}{2k}\right)\mu\bigl( V(\bar{x}(t)) \bigr).
\label{eqn_thm_delayindependent_proof_DV_step3}
\end{align}
Here, (\ref{prooftheorem2-5}) as well as the bound (\ref{eqn_thm_delayindependent_proof_dVi_bound}) are used and it is recalled that $V(\bar{x}(t)) = V(x(t,\varphi))$ by the choice of $\bar{x}(t)$ as (\ref{eqn_thm_delayindependent_proof_xineq}). As (\ref{eqn_thm_delayindependent_proof_DV_step3}) holds for any trajectory that satisfies (\ref{eqn_thm_delayindependent_proof_razumikhincondition}) and $k\geq1$, asymptotic stability of the origin follows from the Razumikhin stability theorem, see \cite[Theorem~7]{driver_1962}, provided that the assumption (\ref{eqn_thm_delayindependent_proof_q_assumption}) holds.

In the final part of the proof, we will construct a function $q$ that satisfies the assumption (\ref{eqn_thm_delayindependent_proof_q_assumption}). To this end, define the compact set
\begin{align*}
\tilde{\Xcal} \defl \{ x\in\R^n_+ \mid 0\leq x \leq 2v\}.
\end{align*}
Clearly, $\Xcal\subset\tilde{\Xcal}$. Since $g$ in (\ref{Delaysystem}) is locally Lipschitz, there is a constant $L_{\tilde{\Xcal}}>0$ such that
\begin{align}
\|g(x,y') - g(x,y)\|_{\infty} \leq L_{\tilde{\Xcal}} \| y' - y \|_{\infty}
\label{eqn_thm_delayindependent_proof_Lipschitz}
\end{align}
for all $x,y,y'\in\tilde{\Xcal}$ and with $\|x\|_{\infty} = \max_i|x_i|$. Since $\bar{x}(t)\in {\Xcal}$, we have $\bar{x}(t)\in \tilde{\Xcal}$. If $\tilde{x}(t)\in \tilde{\Xcal}$, it follows from (\ref{eqn_thm_delayindependent_proof_Lipschitz}) that
\begin{align}
g_j\bigl(\bar{x}(t),\tilde{x}(t)\bigr) - g_j\bigl(\bar{x}(t),\bar{x}(t)\bigr)&\leq L_{\tilde{\Xcal}} \max_{i\in\Ical_n}\;\{ \tilde{x}_i (t)- \bar{x}_i(t)\},
\label{eqn_thm_delayindependent_proof_q_step1}\\
&= L_{\tilde{\Xcal}} \max_{i\in\Ical_n} \;\bigl\{\rho_i\bigl(q(V(\bar{x}(t)))\bigr) - \rho_i\bigl(V(\bar{x}(t))\bigr)\bigr\},
\label{eqn_thm_delayindependent_proof_q_step2}
\end{align}
where the property $\bar{x}(t)\leq\tilde{x}(t)$ and the first inequality in (\ref{eqn_thm_delayindependent_proof_q_assumption}) are used to obtain inequality (\ref{eqn_thm_delayindependent_proof_q_step1}). Equality~\eqref{eqn_thm_delayindependent_proof_q_step2} follows from the definitions (\ref{eqn_thm_delayindependent_proof_xineq}) and (\ref{eqn_thm_delayindependent_proof_xdelayineq}). The desired condition (\ref{eqn_thm_delayindependent_proof_q_assumption}) holds if
\begin{align}
\rho_i\bigl(q(V(\bar{x}(t)))\bigr) - \rho_i\bigl(V(\bar{x}(t)))\bigr)
&\leq \frac{\tilde{\mu}(V(\bar{x}(t)))}{2kDL_{\tilde{\Xcal}}}\nonumber\\
&\leq \frac{\mu(V(\bar{x}(t)))}{2kDL_{\tilde{\Xcal}}}
\label{eqn_thm_delayindependent_proof_q_step3}
\end{align}
for all $i\in\Ical_n$. Here, $\tilde{\mu}:[0,s]\rightarrow\R_+$ is a function of class \classK that lower bounds the positive definite function $\mu$. Such a function $\tilde{\mu}$ exists by the fact that $\mu$ is positive definite on the compact set $[0,s]$~\cite{kellett_2014}. At this point, we note that, even though $\bar{x}(t)\in\tilde{\Xcal}$, this does not necessarily hold for $\tilde{x}(t)$. However, the condition (\ref{eqn_thm_delayindependent_proof_q_step3}) implies that
\begin{align*}
\tilde{x}(t) \leq \bar{x}(t) + \frac{\mu(V(\bar{x}(t)))}{2kDL_{\tilde{\Xcal}}} \ones_n \leq v+\frac{\mu(V(\bar{x}(t)))}{2kDL_{\tilde{\Xcal}}}\ones_n ,
\end{align*}
such that choosing $k$ sufficiently large guarantees $\tilde{x}(t)\in\tilde{\Xcal}$. For this choice, the Lipschitz condition (\ref{eqn_thm_delayindependent_proof_Lipschitz}) indeed holds.

After fixing $k\geq1$ as above and denoting $s = V(\bar{x})$, the function $q:[0,\bar{s}]\rightarrow\R_+$ defined as
\begin{align*}
q(s) \defl \min_{i\in\Ical_n} \rho_i^{-1}\!\left( \rho_i(s) + \frac{\tilde{\mu}(s)}{2kDL_{\tilde{\Xcal}}} \right).
\end{align*}
This function satisfies (\ref{eqn_thm_delayindependent_proof_q_step3}) and, hence, (\ref{eqn_thm_delayindependent_proof_q_assumption}). In addition, since the functions $\rho_i$ and $\tilde{\mu}$ are of class $\classK$, $q$ is nondecreasing. Also, it can be observed that $q(s) > s$ for all $s>0$ as required. Therefore, all conditions on $q$ are satisfied. This finalizes the proof of the implication $2)\Rightarrow1)$.

%
%

\bibliographystyle{unsrt}
\bibliography{all}

\end{document}